\documentclass[a4paper,USenglish,numberwithinsect]{lipics-v2018}

\usepackage{microtype}
\usepackage{upgreek}
\usepackage{csquotes} 
\usepackage{enumerate} 
\usepackage{algorithm,bbm}
\usepackage[noend]{algpseudocode}
\usepackage{amsmath,amssymb,amsthm} 
\usepackage{bm} 
\usepackage{amsfonts}
\usepackage{lmodern}
\usepackage[numbers,sort&compress]{natbib}
\usepackage{graphics,graphicx}

\theoremstyle{plain}

\newenvironment{claimproof}[1][\proofname]{\begin{proof}[#1]\renewcommand{\qedsymbol}{\claimqed}}{\end{proof}\renewcommand{\qedsymbol}{\plainqed}}

\newcommand{\NP}{\ensuremath{\mathsf{NP}}\xspace}

\newcommand{\PisNP}{\ensuremath{\mathsf{P=NP}}\xspace}

\newtheorem{claim}[theorem]{Claim}

\newtheorem{conjecture}[theorem]{Conjecture}
\newcommand\TombStone{\rule{.7ex}{1.7ex}}
\renewcommand{\qedsymbol}{\TombStone}
\newcommand{\qedd}{\let\qed\relax\quad\raisebox{-.1ex}{$\qedsymbol$}}
\newcommand{\claimqedd}{\let\qed\relax\quad\raisebox{-.1ex}{$\lrcorner$}}
\usepackage{tikz}
\usetikzlibrary{arrows,decorations.markings, shapes.misc, shapes.multipart, positioning}

\title{Hamiltonicity below Dirac's condition}

\author{Bart M.P. Jansen}{Eindhoven University of Technology}{b.m.p.jansen@tue.nl}{http://orcid.org/0000-0001-8204-1268}{Supported by NWO Gravitation grant ``Networks''.}
\author{L\'{a}szl\'{o} Kozma}{Freie Universit\"at Berlin}{laszlo.kozma@fu-berlin.de}{}{Supported by ERC Consolidator Grant No 617951.}
\author{Jesper Nederlof}{Eindhoven University of Technology}{j.nederlof@tue.nl}{}{Supported by NWO Gravitation grant ``Networks'' and NWO Grant No 639.021.438.}

\authorrunning{Bart M.P. Jansen, L\'{a}szl\'{o} Kozma, and Jesper Nederlof}

\Copyright{Bart M.P. Jansen, L\'{a}szl\'{o} Kozma, and Jesper Nederlof}

\subjclass{%
,
}


\category{}

\relatedversion{}

\supplement{}

\funding{}

\EventEditors{John Q. Open and Joan R. Access}
\EventNoEds{2}
\EventLongTitle{..}
\EventShortTitle{2019}
\EventAcronym{..}
\EventYear{2018}
\EventDate{..}
\EventLocation{..}
\EventLogo{}
\SeriesVolume{42}
\ArticleNo{23}
\nolinenumbers 
\hideLIPIcs  

\usepackage{xspace}

\begin{document}
\date{}
\maketitle

\begin{abstract}
Dirac's theorem (1952) is a classical result of graph theory, stating that an $n$-vertex graph ($n \geq 3$) is Hamiltonian if every vertex has degree at least $n/2$. Both the value $n/2$ and the requirement for \emph{every vertex} to have high degree are necessary for the theorem to hold.

In this work we give efficient algorithms for determining Hamiltonicity when either of the two conditions are relaxed. More precisely, we show that the Hamiltonian cycle problem can be solved in time $c^k \cdot n^{O(1)}$, for some fixed constant $c$, if at least $n-k$ vertices have degree at least $n/2$, or if all vertices have degree at least $n/2 - k$. The running time is, in both cases, asymptotically optimal, under the exponential-time hypothesis (ETH). 

The results extend the range of tractability of the Hamiltonian cycle problem, showing that it is fixed-parameter tractable when parameterized below a natural bound. 
In addition, for the first parameterization we show that a kernel with $O(k)$ vertices can be found in polynomial time.

\end{abstract}

\section{Introduction}\label{sec:intro}

The {\sc Hamiltonian Cycle} problem asks whether a given undirected graph has a cycle that visits each vertex exactly once. It is a central problem of graph theory, operations research, and computer science, with an early history that well predates these fields (see e.g.\ \cite{Knuth2018}).   

Several conditions that guarantee the existence of a Hamiltonian cycle in a graph are known. Perhaps best known among these is Dirac's theorem from 1952~\cite{Dirac}. It states that a graph with $n$ vertices ($n \geq 3$) is Hamiltonian if every vertex has degree at least $n/2$. Various extensions and refinements of Dirac's theorem have been obtained, often involving further graph parameters besides minimum degree (see e.g.\ the book chapters~\cite[\S\,10]{Diestel}, \cite[\S\,11]{Lawler1985} and survey articles~\cite{gould2014recent, LiSurvey, kuhn2014hamilton} for an overview). We remark that a polynomial-time verifiable condition for Hamiltonicity cannot be both necessary and sufficient, unless \PisNP~\cite{Karp1972}. In its stated form, Dirac's theorem is as strong as possible. In particular, if we replace $n/2$ by $\lfloor n/2 \rfloor$, the graph may fail to be two-connected---a precondition for Hamiltonicity. (Consider two $\lceil n/2 \rceil$-cliques with a common vertex.) 

In this paper we relax the conditions of Dirac's theorem and consider input graphs in which (1) at least $n-k$ vertices have degree at least $n/2$ (the degrees of the remaining vertices can be arbitrarily small), or (2) all vertices have degree at least $n/2 - k$. 

For both relaxations we show that {\sc Hamiltonian Cycle} can be solved deterministically, in time $c^k \cdot n^{O(1)}$, for some fixed constant $c$. This establishes the fixed-parameter tractability of {\sc Hamiltonian Cycle} when parameterized by the distance from Dirac's bound, for two natural ways of measuring this distance. 

The known exact algorithms for {\sc Hamiltonian Cycle} in general graphs have exponential running time (the problem is one of the original $21$ \NP-hard problems~\cite{Karp1972}). The best deterministic running time of $O(2^n \cdot n^2)$ is achieved by the dynamic programming algorithm of Bellman~\cite{Bellman1962}, and Held and Karp~\cite{HeldKarp1962}, and has not been improved since the 1960s. Among randomized algorithms, the current best running time of $O(1.657^n)$ is achieved by the more recent algorithm of Bj\"orklund~\cite{Bjorklund} based on determinants. Improving these bounds remains a central open question of the field. 

Assuming the exponential-time hypothesis (ETH)~\cite{ImpagliazzoEtAl2001}, there is no algorithm for {\sc Hamiltonian Cycle} with running time $2^{o(n)}$. In both parameterizations considered in this paper, $k \leq n$ holds. Thus, under ETH, a running time of the form $2^{o(k)} \cdot n^{O(1)}$ is ruled out, and our algorithms are optimal, up to the base of the exponential.
 Furthermore, there exists a fixed constant $\upalpha>0$, such that our parameterized bounds asymptotically improve the current best bounds for {\sc Hamiltonian Cycle}, if the value of $k$ is at most $\upalpha \cdot n$. 

For the first parameterization, we show that {\sc Hamiltonian Cycle} admits a kernel with $O(k)$ vertices, computable in polynomial time. In other words, the input graph can be compressed (roughly) to the order of its sparse part, while preserving Hamiltonicity.

Our results show that checking Hamiltonicity becomes tractable as we approach the degree-bound of Dirac's theorem. The crude intuition behind Dirac's theorem (and many of its generalizations) is that \emph{having many edges} makes a graph Hamiltonian. It is a priori far less obvious why approaching the Dirac bound would make the \emph{algorithmic problem} easier; one may even expect that the more edges there are, the harder it becomes to certify \emph{non-Hamiltonicity}. To provide some intuition why this is not the case, we give a brief informal summary of the arguments.

When $n-k$ vertices have degree at least $n/2$, i.e.\ in the first case, our algorithm takes advantage of the fact that, by a result of Bondy and Chv\'atal, the subgraph induced by the high-degree vertices can be completed to a clique without changing the Hamiltonicity of the graph; all relevant structure is thus in the sparse part and its interconnection with the dense part. Then, we find a subset of the vertices in the clique that are well-connected to the sparse part (by solving a matching problem in an auxiliary graph), and we ignore the remainder of the clique. Finally, we show how a Hamiltonian cycle on this smaller, well-connected subgraph, can be extended to a Hamiltonian cycle of the entire graph, guided by the alternating paths of the matching. For this parameterization we are not aware of a comparable result in the literature. 

When all vertices have degree at least $n/2 - k$, i.e.\ in the second case, a result of Nash-Williams implies that either a Hamiltonian cycle, or a sufficiently large independent set can be found in polynomial time. In the latter case, we certify non-Hamiltonicity by showing (roughly) that the complement of the independent set is not coverable by a certain number of disjoint paths. This argument is essentially the same as the one given by H\"aggkvist~\cite{haggkvist} towards his algorithm with running time $O(n^{5k})$ for the same parameterization. (H\"aggkvist states this algorithmic result as a corollary of structural theorems. He does not describe the details of the algorithm or its analysis, but these are not hard to reconstruct.) Here we improve the running time of H\"aggkvist's algorithm to the stated (asymptotically optimal) $c^k \cdot n^{O(1)}$ by more efficiently solving the arising path-cover subproblem.

\subsection{Statement of results}

Our first result shows that if a graph has a ``relaxed'' Dirac property, it can be compressed while preserving its Hamiltonicity.

\begin{theorem}\label{thm1}
Let $G$ be an $n$-vertex graph such that at least $n-k$ vertices of $G$ have degree at least $n/2$. There is a deterministic algorithm that, given $G$, constructs in time $O(n^3)$ a $3k$-vertex graph $G'$, such that $G$ is Hamiltonian if and only if $G'$ is Hamiltonian. 
\end{theorem}

Equivalently stated in the language of parameterized complexity, the Hamiltonian cycle problem parameterized by $k$ has a \emph{kernel} with a linear number of vertices. To determine the Hamiltonicity of a graph $G$, we simply apply the algorithm of Theorem~\ref{thm1} to compress $G$, and use an exponential-time algorithm (for instance, the Held-Karp algorithm) to solve {\sc Hamiltonian Cycle} directly on the compressed graph. We thus obtain the following result. 

\begin{corollary}
If at least $n-k$ vertices of an $n$-vertex graph $G$ have degree at least $n/2$, then {\sc Hamiltonian Cycle} with input $G$ can be solved in deterministic time $O(8^k \cdot k^2 + n^3)$. 
\end{corollary}

As an alternative, we may also use an approach based on inclusion-exclusion~\cite{KarpIE} to solve the reduced {\sc Hamiltonian cycle} instance, achieving the overall running time $O(8^k \cdot k^3 + n^3)$, with \emph{polynomial space}.

Our result for the second relaxation of Dirac's theorem is as follows.

\begin{theorem}\label{thm2}
If every vertex of an $n$-vertex graph $G$ has degree at least $n/2 - k$, then {\sc Hamiltonian Cycle} with input $G$ can be solved in deterministic time $O(30^{6k} \cdot n^3)$. 
\end{theorem}

The running time of the Bellman-Held-Karp algorithm for {\sc Hamiltonian Cycle} is $O(2^n \cdot n^2)$. Denoting $\upalpha = k/n$, our results represent an asymptotic improvement if $\upalpha < 1/3$ in the first parameterization, and if $\upalpha < 0.0339$ in the second parameterization. 

As a counterpoint to our results, we mention that {\sc Hamiltonian Cycle} remains hard (in both parameterizations) for arbitrarily small values of $\upalpha$.

\begin{theorem}
Assuming ETH, {\sc Hamiltonian Cycle} cannot be solved in time $2^{o(n)}$ in $n$-vertex graphs with at least $(1 - \upalpha) \cdot n$ vertices of degree at least $n/2$, and in $n$-vertex graphs with minimum degree $(1 - \upalpha) \cdot n/2$, for arbitrary fixed $0 < \upalpha < 1/2$. 
\end{theorem}
\begin{proof}
In both cases we construct a graph with the given degree-requirements that embeds a hard instance of {\sc Hamiltonian Path} with $\upalpha \cdot n$ vertices. 
For the second statement we can use the construction from the NP-hardness proof of Dahlhaus, Hajnal, and Karpinski~\cite{dahlhaus}. For the first statement, consider an $\upalpha \cdot n$-vertex instance of {\sc Hamiltonian Path}, connected by two disjoint edges to an $(1-\upalpha) \cdot n$-vertex clique. \qedd
\end{proof}

\subsection{Related work}

In general, parameterized complexity~\cite{FG06, book4} allows a finer-grained understanding of algorithmic problems than classical, univariate complexity. No new insight is gained, however, if the chosen parameter $k$ is large in all interesting cases. For example, in planar graphs, the Four Color Theorem guarantees the existence of an independent set of size $n/4$. As a consequence, any exponential-time algorithm for maximum independent set trivially achieves fixed-parameter tractability in terms of the solution size. 

To deal with this issue, Mahajan and Raman~\cite{MahajanRaman} introduced the method of parameterizing problems \emph{above} or \emph{below} a guaranteed bound. (Similar considerations motivate the ``distance from triviality'' framework of Guo, H\"uffner, and Niedermeier~\cite{Guo}.) In the example of planar independent set, an interesting parameter is the amount by which the solution size exceeds $n/4$. Similar ideas have successfully been applied to several problems (see e.g.\ \cite{MRS, IPEC09, Alon2011, EE, vertex_cover_param, BCDF}). Our results also fall in the framework of ``above/below'' parameterization, with the remark that our parameter of interest is not the value to be optimized but a structural property of the input, which we parameterize near its ``critical value''. 

Perhaps closest to our work is the recent result of Gutin and Patel~\cite{TSPavg} on the Traveling Salesman problem, parameterized below the cost of the \emph{average} tour. Although it concerns Hamiltonian cycles (in an edge-weighted complete graph), the result of Gutin and Patel is not directly comparable with our results. In particular, averaging arguments do not seem to help when studying the \emph{existence} of Hamiltonian cycles, which is often determined by local structure in the graph. For instance, {\sc Hamiltonian Cycle} remains \NP-hard even in graphs with average degree $\upalpha n$ for any constant $\upalpha < 1$. (Consider a clique of $\sqrt{\upalpha} n$ vertices, connected by two non-incident edges to the remaining graph that encodes a hard instance of {\sc Hamiltonian Path}.)

\section{Preliminaries}

We use standard graph-theoretic notation (see e.g.~\cite{Diestel}). An edge between vertices $u$ and $v$ is written simply as $uv$ or $vu$. The \emph{neighborhood} of a vertex $v$ in graph $G$ is denoted by $N_G(x)$. The \emph{degree} of $v$ in $G$ is $d_G(v) = |N_G(v)|$, and the minimum degree of $G$ is $\delta_G = \min_{v \in V(G)}{d_G(v)}$. We conveniently omit the subscript $G$ whenever possible. For a set $S \subseteq V(G)$ of vertices, $G[S]$ denotes the subgraph induced by $S$ on $G$.

We state Dirac's theorem and a strengthened statement due to Ore. Let $G$ be an $n$-vertex undirected graph, with $n \geq 3$.

\begin{lemma}[Dirac~\cite{Dirac}]\label{lemdirac}
If $\delta \geq n/2$, then $G$ is Hamiltonian.
\end{lemma}

\begin{lemma}[Ore~\cite{Ore}]\label{lemore}
If $d(u) + d(v) \geq n$ for every non-adjacent pair of  vertices $u$, $v$ of $G$, then $G$ is Hamiltonian.
\end{lemma}

We state a theorem of Bondy and Chv\'atal that we use in the proofs of both Theorem~\ref{thm1} and Theorem~\ref{thm2}.

\begin{lemma}[Bondy-Chv\'atal~\cite{BondyChvatal}]\label{bondy-chvatal}
Let $G$ be an $n$-vertex graph, and let $G'$ be obtained from $G$ by adding an edge $uv$ to $G$ for some pair of non-adjacent vertices $u,v$ such that $d_G(u) + d_G(v) \geq n$. Then $G'$ is Hamiltonian if and only if $G$ is Hamiltonian. Moreover, given a Hamiltonian cycle of $G'$, a Hamiltonian cycle of $G$ can be obtained in linear time.
\end{lemma}

It is easy to see that Lemma~\ref{bondy-chvatal} implies both Lemma~\ref{lemdirac} and Lemma~\ref{lemore}, as in both cases we can iterate the edge-augmentation step until obtaining a complete graph.

Finally, we state yet another strengthening of Dirac's theorem, due to Nash-Williams~\cite{nw}. We write this result in a slightly non-standard, explicitly algorithmic form. Our use of this result in proving Theorem~\ref{thm2} is the same as in the argument of H\"aggkvist~\cite{haggkvist}.

\begin{lemma}[Nash-Williams~\cite{nw}]\label{lemnw}
Let $G$ be a $2$-connected graph with $n$ vertices, with $\delta \geq (n+2)/3$. Then, we can find in $G$, in time $O(n^3)$, either a Hamiltonian cycle, or an independent set of size $\delta+1$.
\end{lemma}

The following proof of Lemma~\ref{lemnw} is due to Bondy~\cite{bondy}, sketched in~\cite[\S\,11]{Lawler1985}. We spell it out fully to make our discussion self-contained and to provide an explicitly algorithmic form (this requires only minor changes compared to the presentation in~\cite{Lawler1985}). 

\begin{proof}

All cycles considered in the proof are simple. Denote $V = V(G)$. Start with an arbitrary cycle $C$ of $G$ (the fact that $G$ is $2$-connected guarantees the existence of a cycle, and we can easily find one in linear time). We extend $C$ into successively longer cycles until we either (1) reach a Hamiltonian cycle, or (2) find an independent set of the required size. 

Unless we have already found a Hamiltonian cycle, $|C| \leq n-1$ holds. Suppose $V \setminus C$ is an independent set, and let $v \in V \setminus C$ be an arbitrary vertex. Due to the independence of $V \setminus C$, we have $N(v) \subseteq C$. If two neighbors of $v$ are connected by an edge of $C$, then we can immediately extend $C$ via $v$. Assume therefore, that this is not the case. Fix an arbitrary orientation of $C$, and let $N^{+}$ be the set of successors in $C$ of the vertices $N(v)$. Then, $|N^{+}| \geq \delta$. Again, if two vertices $x,y \in N^{+}$ are connected, then $C$ can be further extended (see Figure~\ref{fig1}(a) for illustration). Thus, we can assume that $N^{+} \cup \{v\}$ is an independent set of $G$, of the required size.

It remains to show that $C$ can be extended whenever $V \setminus C$ is \emph{not} an independent set. Then, $|V \setminus C| \geq 2$ must hold. Let us orient $C$ arbitrarily and label its vertices accordingly as $x_1, \dots, x_k$. Let $P = (x_1, p_1, p_2, \dots, x_{t+1})$, with $1 \leq t < k$, be a simple path of length at least $3$, that intersects $C$ only at the endpoints $x_1$ and $x_{t+1}$ (see Figure~\ref{fig1}(c)). We claim that the existence of such a path can be assumed without loss of generality (by suitably choosing the starting label $x_1$).

To see this, consider a path $(u,v,w)$ in $G$, where $u \in C$ and $v,w \in V \setminus C$. (Such a path must exist by the assumption that $V \setminus C$ is not independent and the fact that $G$ is $2$-connected: start with an arbitrary edge outside $C$ and consider a path from one of its endpoints to a vertex of $C$).

If there is a path, vertex-disjoint from $\{u,v\}$, from $w$ to an arbitrary vertex $x \in C \setminus \{u\}$, then we obtain the desired structure by labeling $x_1 = u$, $p_1 = v$, $p_2 = w$, and $x_{t+1} = x$. 

If there is no such path, then there must be a path $R$, internally vertex-disjoint from $\{u,v\}$, from $v$ to an arbitrary vertex $x \in C \setminus \{u\}$ (otherwise, deleting $u$ would disconnect $G$, contradicting its $2$-connectivity). Furthermore, there must exist a path $Q$ connecting $u$ to $w$, not containing $v$ (otherwise, deleting $v$ would separate $u$ and $w$) and internally vertex-disjoint from $C$ and $R$ (otherwise, we would have a path from $w$ to a vertex in $C \setminus \{u\}$, ruled out previously). Now, we obtain the desired structure by setting $x_1 = u$, $x_{t+1} = x$, and $p_1$ and $p_2$ the second and third vertices on the path consisting of $Q$, and the edge $wv$ (Figure~\ref{fig1}(b)).

Let $d^{\uparrow}_t$, $d^{\uparrow}_k$, $d^{\uparrow}_2$ denote the number of neighbors $x_i$ with $1 \leq i \leq t$ of $x_t$, $x_k$, resp.\ $p_2$. Similarly, let $d^{\downarrow}_t$, $d^{\downarrow}_k$, $d^{\downarrow}_2$ denote the number of neighbors $x_i$ with $t < i \leq k$ of $x_t$, $x_k$, resp.\ $p_2$. Let $d^{\circ}_t$, $d^{\circ}_k$, $d^{\circ}_2$ denote the number of neighbors in $V \setminus C$ of $x_t$, $x_k$, resp.\ $p_2$.

\begin{claim}
At least one of the following three inequalities holds:
\begin{enumerate}[(1)]
\item $d^{\uparrow}_t + d^{\uparrow}_k + d^{\uparrow}_2> t+1$, 
\item $d^{\downarrow}_t + d^{\downarrow}_k + d^{\downarrow}_2 > k-t+1$.
\item $d^{\circ}_t + d^{\circ}_k + d^{\circ}_2 > n-k-1$

\end{enumerate}
\end{claim}
\begin{claimproof}
Suppose otherwise. Then the sum of degrees of $x_t$, $x_k$, $p_2$ is at most $(t + 1) + (k-t+1) + (n-k-1) \leq n+1$, and thus at least one of them has degree at most $(n+1)/3 < \delta$, a contradiction. 
\end{claimproof}

In the following, we assume $d^{\uparrow}_2 \leq 1$, since, if $p_2$ is connected to some $x_i$ with $1 < i \leq t$, then we may choose a different index $t$ in our construction.

Suppose inequality (1) holds. Then, $d^{\uparrow}_t + d^{\uparrow}_k > t$, and by the pigeonhole-principle, there is some $1 \leq i < t$ such that $x_k x_{i+1}$ and $x_{t} x_i$ are edges of $G$. Then, $C$ can be extended by adding these two edges and the path $(x_1,p_1,\dots,x_{t+1})$ and removing the edges $x_k x_1$, $x_t x_{t+1}$, and $x_i x_{i+1}$. (See Figure~\ref{fig1}(c).)

Suppose inequality (2) holds. Then, there is some $t+1 \leq i < k$ such that one of the following is true: 
(a) $x_k x_{i}$ and $x_{t} x_{i+1}$ are edges of $G$, 
(b) $x_t x_{i}$ and $p_{2} x_{i+1}$ are edges of $G$, 
(c) $i < k-1$, and $x_k x_{i}$ and $p_{2} x_{i+2}$ are edges of $G$,
(d) $x_k x_{t}$ or $p_{2} x_{k}$ is an edge of $G$. 
In all these cases $C$ can be extended similarly to the previous case. (See Figure~\ref{fig1}(d) for an illustration of cases (a), (b), (c).) 

To see that one of the four cases must hold, we can argue by contradiction. Apart from the boundaries, every vertex $x_i$ that is connected to $x_t$ rules out $x_{i-1}$ from being connected to $x_k$. Similarly, every vertex $x_i$ that is connected to $p_2$ rules out $x_{i-2}$ from being connected to $x_k$, as well as $x_{i-1}$ from being connected to $x_t$. Thus, fixing the connections from $x_t$ and $p_t$ first, we rule out $d^{\downarrow}_t + d^{\downarrow}_k - 2$ possible neighbors of $x_k$. 

Suppose inequality (3) holds. Then, by a similar pigeonhole-argument, one of the following is true: 
(a) $x_t p_j$ or $x_k p_j$ is an edge of $G$ for some internal vertex $p_j$ on path $P$, or
(b) at least two of $x_t w$, $x_k w$, and $p_2 w$ are edges of $G$, for some $w \in V \setminus (C \cup P)$.
In all these cases $C$ can be extended similarly to the previous cases.

The claimed running time can be achieved via a straightforward implementation. \qedd
\end{proof}

\begin{figure}[h]
  \centering
  \includegraphics[width=14cm]{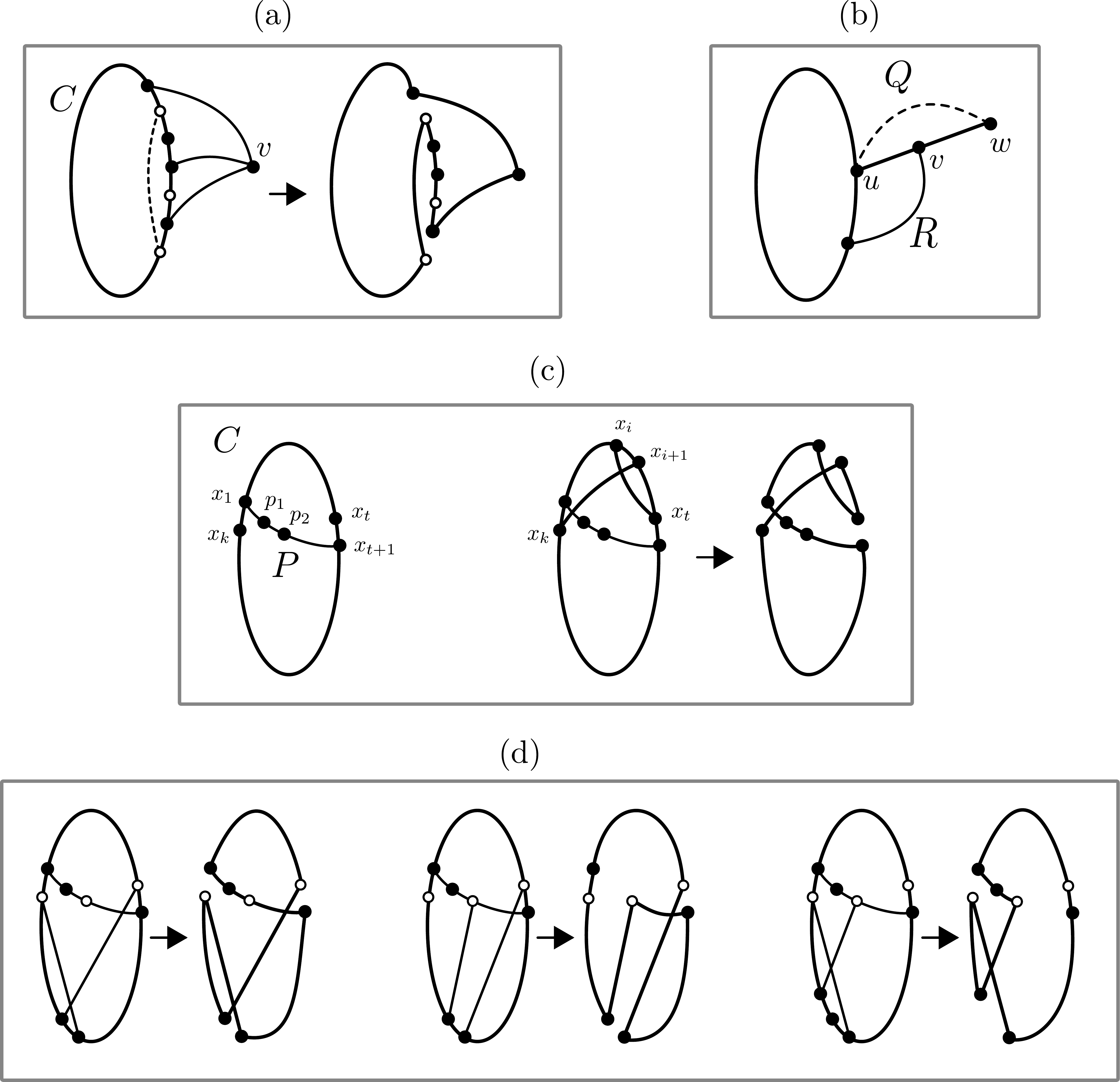}
  \caption{Illustration of the proof of Lemma~\ref{lemnw}. (a) cycle $C$ with sets $N(v)$ and $N^{+}$ of vertices shown as solid, resp.\ hollow circles; (b) constructing a path between two vertices of $C$; (c) one case of extending $C$; (d) extending $C$ in the first three cases, when inequality (2) holds. Vertices $x_t$, $x_k$, $p_2$ shown as hollow dots.\label{fig1}}
\end{figure}

\section{Relaxing the cardinality-constraint (proof of Theorem~\ref{thm1})}\label{sec2}

Let $C \subseteq V(G)$ denote the set of high-degree vertices of $G$ (those with degree at least $n/2$), and let $S = V(G) \setminus C$ denote the remaining (i.e.\ low-degree) vertices.

Observe that $|S| \leq k$. By Lemma~\ref{bondy-chvatal}, we may add all edges between vertices in $C$, without changing the Hamiltonicity of $G$, assume therefore that $C$ is a clique. 

The proof of the following theorem is inspired by the \emph{crown reductions}~\cite{Abu-KhzamFLS07,ChorFJ04,Fellows03} used to obtain kernels for \textsc{Vertex Cover} and \textsc{Saving $k$ Colors}.

\begin{theorem} \label{thm:dist:to:clique}
There is a polynomial-time algorithm that, given a graph~$G$ and a nonempty set~$S \subseteq V(G)$ such that~$G-S$ is a clique, outputs an induced subgraph~$G'$ of~$G$ on at most~$3|S|$ vertices such that~$G$ is Hamiltonian if and only if~$G'$ is Hamiltonian.
\end{theorem}
\begin{proof}
Given a graph~$G$ let $S \subseteq V$, such that $C := V(G) \setminus S$ is the vertex set of a clique in $G$. If~$C \leq 2|S|$ then~$G' := G$ suffices, so we assume~$C > 2|S|$ in the remainder. Let~$S' := \{v_1, v_2 \mid v \in S\}$ be a set containing two representatives for each vertex of~$S$. Construct a bipartite graph~$H$ on vertex set~$C \cup S'$. For each edge~$cv \in E(G)$ with~$c \in C$ and~$v \in S$, add the edges~$c v_1, c v_2$ to~$H$. Compute a maximum matching~$M \subseteq E(H)$ in graph~$H$, for example using the Edmonds-Karp algorithm. Let~$C^*$ be the vertices of~$C$ saturated (matched) by~$M$. If~$|C^*| \geq |S|+1$ then let~$C' := C^*$, and otherwise let~$C' \subseteq C$ be a superset of~$C^*$ of size~$|S| + 1$. Output the graph~$G' := G[C' \cup S]$ as the result of the reduction.

\begin{claim} \label{claim:size}
Graph~$G'$ has at most~$3|S|$ vertices.
\end{claim}
\begin{claimproof}
Since each vertex of~$C^*$ is matched to a distinct vertex in~$S'$, with~$|S'| = 2|S|$, it follows that~$|C^*| \leq 2|S|$ which implies~$|C'| \leq 2|S|$. As~$V(G') = C' \cup S$, the claim follows. 
\end{claimproof}

The output graph~$G'$ therefore satisfies the size bound. It remains to prove that it is equivalent to~$G$ with respect to Hamiltonicity. We first prove the simpler implication.

\begin{claim} \label{claim:prime:ham:g:ham}
If~$G'$ is Hamiltonian, then~$G$ is Hamiltonian.
\end{claim}
\begin{claimproof}
Suppose that~$G'$ is Hamiltonian, and let~$F \subseteq E(G)$ be a Hamiltonian cycle in~$G'$. Fix an arbitrary orientation of~$F$. As each vertex from~$C'$ has a unique successor on~$F$, while~$|C'| > |S|$ by definition, it follows that some vertex~$x \in C'$ has a successor from~$C'$ along the cycle; let this be~$y \in C'$. Then we can transform~$F$ into a Hamiltonian cycle in~$G$ by removing the edge~$xy$ and replacing it by a path through all the clique-vertices of~$C \setminus C'$. 
\end{claimproof}

The remainder of the proof is aimed at proving the reverse implication. For this, we introduce some terminology. For a vertex set~$S^*$ in a graph~$G^*$, we define a \emph{path cover of~$S^*$ in~$G^*$} as a set of pairwise vertex-disjoint simple paths~$P_1, \ldots, P_\ell$ in~$G^*$, such that each vertex of~$S^*$ belongs to exactly one path~$P_i$. For a vertex set~$C^*$ in~$G^*$, we say the path cover \emph{has $C^*$-endpoints} if the endpoints of each path~$P_i$ belong to~$C^*$. We will sometimes interpret a subgraph in which each connected component is a path as a path cover, in the natural way.

\begin{claim} \label{claim:pathcover:to:cycle}
If there is a path cover of~$S$ in~$G'$ having~$C'$-endpoints, then~$G'$ is Hamiltonian.
\end{claim}
\begin{claimproof}
Any path cover of~$S$ consists of at least one path (since~$S$ is nonempty by assumption) and the endpoints of the paths are all distinct. Hence a path cover consisting of~$\ell \geq 1$ paths has exactly~$2\ell$ distinct endpoints~$\{s_1, t_1, \ldots, s_\ell, t_\ell\}$, which are vertices in the clique~$C'$. Let~$P_{\ell+1}$ be a simple path in~$G'$ visiting all vertices that are not touched by the path cover; such a path exists because the only vertices not touched by the path cover belong to the clique~$C'$. Then one can obtain a Hamiltonian cycle in~$G'$ by taking the edges of~$P_1, \ldots, P_\ell, P_{\ell+1}$, together with edges connecting the end of path~$P_i$ to the beginning of path~$P_{i+1}$ for all relevant values of~$i$. 
\end{claimproof}

To prove that Hamiltonicity of~$G$ implies Hamiltonicity of~$G'$, we will construct a path cover of~$S$ in~$G'$ having~$C'$-endpoints, using a hypothetical Hamiltonian cycle in~$G$. To do so we need several properties enforced by the matching~$M$ in~$H$, which we now explore.

Let~$U_C$ be the vertices of~$C$ that are not saturated by~$M$. Let~$R$ denote the vertices of~$H$ that are reachable from~$U_C$ by an $M$-alternating path in the bipartite graph $H$ (which necessarily starts with a non-matching edge), and define~$R_C := R \cap C$ and~$R_{S'} := R \cap S'$.

\begin{claim} \label{claim:properties:matching}
The sets~$R, R_C, R_{S'}$ satisfy the following.
\begin{enumerate}
	\item Each $M$-alternating path in~$H$ from~$U_C$ to a vertex in~$R_{S'}$ (resp.~$R_C$) ends with a non-matching (resp.~matching) edge.\label{eq:last:edge}
	\item Each vertex of~$R_{S'}$ is matched by~$M$ to a vertex in~$R_{C}$. \label{eq:sprime:saturated}
	\item For each vertex~$x \in R_C$ we have~$N_H(x) \subseteq R_{S'}$. \label{eq:neighborsrc:reachable}
	\item For each vertex~$v \in S$ we have~$v_1 \in R_{S'} \Leftrightarrow v_2 \in R_{S'}$.\label{eq:copies:behave:alike}
	\item For each vertex~$v \in S' \setminus R_{S'}$, we have~$N_H(v) \cap R_{C} = \emptyset$ and each vertex of~$N_H(v)$ is saturated by~$M$.\label{eq:neighbors:notreachable}
\end{enumerate}
\end{claim}
\begin{claimproof}
\eqref{eq:last:edge} An $M$-alternating path starting in~$U_C$ must start with a non-matching edge, since~$U_C$ consists of unsaturated vertices, and it starts from the $C$-partite set of~$H$. Hence such a path moves to the~$S'$-partite set over non-matching edges, and moves back to the $C$-partite set over matching edges.

\eqref{eq:sprime:saturated} If a vertex~$x \in R_{S'} \subseteq R$ is not saturated, then the $M$-alternating path from~$U_C$ witnessing~$x \in R$ starts and ends with a non-matching edge (by \eqref{eq:last:edge}) and is in fact an $M$-augmenting path. This contradicts that~$M$ is a maximum matching. Hence each~$x \in R_{S'}$ is matched by~$M$ to some vertex~$y$. By \eqref{eq:last:edge} the $M$-alternating path from~$U_C$ to~$x$ that witnesses~$x \in R_{S'}$ ends with a non-matching edge, so together with the matching edge~$\{x,y\}$ this forms an $M$-alternating path witnessing~$y \in R_C$.

\eqref{eq:neighborsrc:reachable} 
Consider a vertex~$x \in R_C$ and an $M$-alternating path~$P$ from~$U_C$ witnessing~$x \in R$. By \eqref{eq:last:edge} the last edge on~$P$ (if any) is a matching edge. Hence if~$x$ is saturated by~$M$, then its matching partner~$y$ is the predecessor of~$x$ on~$P$ and a prefix of~$P$ witnesses~$y \in R$ and hence~$y \in R_{S'}$. For any vertex~$z \in N_H(x)$ that is not the matching partner of~$x$, we can augment~$P$ by the edge~$xz$ to obtain an $M$-alternating path from~$U_C$ to~$z$ witnessing~$z \in R_{S'}$. Together, these two arguments show~$N_H(x) \subseteq R_{S'}$.

\eqref{eq:copies:behave:alike}
Suppose~$v_1 \in R_{S'}$ and let~$P$ be an $M$-alternating path from~$U_C$ to~$v_1$. By~\eqref{eq:last:edge} path~$P$ ends with a non-matching edge~$x v_1$. Since~$v_1$ and~$v_2$ have identical neighborhoods in~$H$, we can replace the last edge of~$P$ by~$x v_2$ to obtain an $M$-alternating path witnessing~$v_2 \in R_{S'}$. The case that~$v_2 \in R_{S'}$ is symmetric.

\eqref{eq:neighbors:notreachable}
Consider~$v \in S' \setminus R_{S'}$. If~$x \in N_H(v) \cap R_{C}$, then~\eqref{eq:neighborsrc:reachable} implies~$v \in R_{S'}$, a contradiction. Hence~$N_H(v) \cap R_{C} = \emptyset$. An unsaturated $H$-neighbor~$x$ of~$v$ would imply~$x \in N_H(v) \cap U_C \subseteq N_H(v) \cap R_C$, so each vertex of~$N_H(v)$ is saturated by~$M$. 
\end{claimproof}

Using these structural insights we can now prove the desired converse to Claim~\ref{claim:prime:ham:g:ham}. Before we give the formal proof, we present the main idea. To prove that~$G'$ is Hamiltonian if~$G$ is, we take a Hamiltonian cycle~$F$ in~$G$ and turn it into a path cover of~$S$ in~$G'$ with~$C'$-endpoints. Any Hamiltonian cycle~$F$ in~$G$ yields a path cover of~$S$ with~$S$-endpoints, by simply taking the restriction of~$F$ onto the vertices of~$S$. The challenge is to extend this path cover with edges into~$C'$ to give it the desired $C'$-endpoints: if the Hamiltonian cycle~$F$ used an edge to jump from~$S$ to~$C$, we have to provide a similar jump in~$G'$. If~$F$ jumps from a vertex~$v \in S$ whose corresponding copies~$v_1, v_2 \in S'$ do not belong to~$R_{S'}$, then the $C$-endpoint of the jumping edge is saturated by~$M$, belongs to~$C'$ and therefore to~$G'$, and can be used to provide the analogous jump in~$G'$. On the other hand, for all vertices~$v \in S$ whose copies~$v_1, v_2$ belong to~$R_{S'}$, we will globally assign new jumping edges based on the matching~$H$. The properties of a matching will ensure that these jumping edges lead to distinct targets and give a valid path cover of~$S$ in~$G'$ having~$C'$-endpoints. We now formalize these ideas.

\begin{claim} \label{claim:g:to:prime}
If~$G$ is Hamiltonian, then~$G'$ is Hamiltonian.
\end{claim}
\begin{claimproof}
Let~$F$ be a Hamiltonian cycle in~$G$. By Claim~\ref{claim:pathcover:to:cycle} it suffices to build a path cover of~$S$ in~$G'$ with~$C'$-endpoints. View~$F$ as a 2-regular subgraph of~$G$, and let~$F_1 := F[S]$ be the subgraph of~$F$ induced by~$S$. Since~$F$ spans~$G$ and all vertices of~$S$ are present in~$G'$, it follows that~$F_1$ is a path cover of~$S$ in~$G'$. However, the paths in~$F_1$ have their endpoints in~$S$ rather than in~$C'$. We resolve this issue by inserting edges into~$F_1$ to turn it into an acyclic subgraph~$F_2$ of~$G'$ in which each vertex of~$S$ has degree exactly two. This structure~$F_2$ must be a path cover of~$S$ in~$G'$ with~$C'$-endpoints, since the degree-two vertices~$S$ cannot be endpoints of the paths. To do the augmentation, initialize~$F_2$ as a copy of~$F_1$. Define~$R_S := \{v \in S \mid v_1 \in R_{S'} \vee v_2 \in R_{S'}\}$ and proceed as follows.
\begin{itemize}
	\item For each vertex~$v \in R_S$, we have~$v_1, v_2 \in R_{S'}$ by Claim~\ref{claim:properties:matching}\eqref{eq:copies:behave:alike}, which implies by Claim~\ref{claim:properties:matching}\eqref{eq:sprime:saturated} that both~$v_1$ and~$v_2$ are matched to distinct vertices~$x_1, x_2$ in~$R_C$. If~$v$ has degree zero in subgraph~$F_1$, then add the edges~$vx_1$, $vx_2$ to~$F_2$. If~$v$ has degree one in~$F_2$ then only add the edge~$vx_1$. Do not add any edges if~$v$ already has degree two in~$F_1$.
	\item For each vertex~$v \in S \setminus R_S$, we claim that~$N_G(v) \cap R_C = \emptyset$. This follows from the fact that~$N_G(v) = N_H(v_1) = N_H(v_2)$ and~Claim~\ref{claim:properties:matching}\eqref{eq:neighbors:notreachable}, using that~$v \notin R_S$ implies~$v_1, v_2 \notin R_{S'}$. Hence the (up to two) neighbors that~$v \in S \setminus R_S$ has in~$C$ on the Hamiltonian cycle~$F$ do not belong to~$R_C$, while~Claim~\ref{claim:properties:matching}\eqref{eq:neighbors:notreachable} ensures that all vertices of~$N_G(v)$ are saturated by~$H$ and hence belong to~$C'$. For each vertex~$v \in S \setminus R_S$, for each edge from~$v$ to~$C \cap C'$ incident on~$v$ in~$F$, we insert the corresponding edge into~$F_2$. 
\end{itemize}

It is clear that the above procedure produces a subgraph~$F_2$ in which all vertices of~$S$ have degree exactly two. To see that~$F_2$ is indeed a path cover, having no vertex of degree larger than two, it suffices to notice that the edges inserted for~$v \in R_S$ connect to \emph{distinct} vertices in~$C' \cap R_C$, while the edges inserted for~$v \in S \setminus R_S$ connect to~$C' \setminus R_C$ in the same way as in the Hamiltonian cycle~$F$. Hence~$F_2$ forms a path cover of~$S$ in~$G'$ having~$C'$-endpoints, which implies that~$G'$ is Hamiltonian and proves Claim~\ref{claim:g:to:prime}. 
\end{claimproof}
Claims~\ref{claim:prime:ham:g:ham} and~\ref{claim:g:to:prime} prove the correctness of the reduction and Claim~\ref{claim:size} gives the desired size bound. Since the reduction can easily be performed in polynomial time, this completes the proof of Theorem~\ref{thm:dist:to:clique}. \qedd
\end{proof}

Observe that the proof of Lemma~\ref{claim:prime:ham:g:ham} explicitly constructs the Hamiltonian cycle in case of a ``yes''-answer. The running time of the reduction is dominated by the bipartite matching step, and the process of undoing the Bondy-Chv\'atal augmentations (Lemma~\ref{bondy-chvatal}), if a cycle of the original graph is to be constructed. Both tasks can be performed in time $O(n^3)$.

\section{Relaxing the degree-constraint (proof of Theorem~\ref{thm2})}\label{sec3}

The outline of the proof largely follows an earlier argument of H\"aggkvist~\cite{haggkvist}. We improve the $O(n^{5k})$ running time of H\"aggkvist's algorithm to $c^k \cdot n^{O(1)}$. 

The algorithm either finds a Hamiltonian cycle or constructs a certificate of non-Hamiltonicity, in the form of a cut $(S,T)$ of the graph, such that the vertices of $T$ can not be covered by $|S|$ vertex-disjoint paths, and this certificate can be verified within the required running time. (Observe that a Hamiltonian cycle induces such a path-cover for an arbitrary cut; paths consisting of single vertices are allowed.)

Assume that $k < n/34$, and thus $\delta > 8n/17$. (Otherwise we revert to a standard exponential-time algorithm.) Furthermore, $\delta < n/2$ may be assumed, as otherwise $G$ is Hamiltonian by Dirac's theorem. Also assume that $G$ is $2$-connected (otherwise it is not Hamiltonian).

Start by running the procedure from the proof of Lemma~\ref{lemnw}, either obtaining a Hamiltonian cycle, or an independent set of size $\delta + 1$. Assume that the latter is the case, and label the obtained independent set as $A_1$.

Partition $V(G)$ into sets $A_1$, $A_2$, and $A_3$, where $A_2$ denotes the set of vertices in $v \in V(G) \setminus A_1$ such that $|N(v) \cap A_1| \geq \delta/2$, and $A_3 = V(G) \setminus (A_1 \cup A_2)$. In words, $A_2$ contains vertices that are sufficiently highly connected to the obtained independent set, and $A_3$ contains the remaining vertices. 

\begin{lemma}[\cite{haggkvist}, Thm.\ 2] \label{lemhaggkvist}
Given sets $A_1$, $A_2$, $A_3$ as defined, we can find a set of vertices $S \subseteq V(G)$ such that $|S| \geq 3 \delta -n +2$, and $G[V(G) \setminus S]$ can be covered by $|S|$ vertex-disjoint paths if and only if $G$ is Hamiltonian. 
\end{lemma}

We sketch the argument, referring to H\"aggkvist~\cite[p.\ 32-33]{haggkvist} for the full details.

\begin{proof}[Proof sketch]
Let $T = V(G) \setminus S$. The ``if'' direction is trivial, since every Hamiltonian cycle of $G$ induces a cover of $G[T]$ by $|S|$ vertex-disjoint paths, for arbitrary $S$.

It remains to show the converse, i.e.\ that a cover of $G[T]$ by $|S|$ vertex-disjoint paths can be extended into a Hamiltonian cycle of $G$, for a suitably chosen $S$.

We let $S=A_2$ if $|A_2| \leq |A_1 \cup A_3|$, and $S = A_1 \cup A_3$ otherwise. 

Consider the bipartite graph with sides $S$ and $T$ containing only those edges of $G$ that have one endpoint in $S$ and another endpoint in $T$. Add to this bipartite graph the edges of the path-cover of $T$ by $|S|$ vertex-disjoint paths (assuming such a cover exists). Call this graph $B$.
Let $B'$ be the graph obtained from $B$ by connecting all pairs of vertices in $S$.

Observe that $B$ is Hamiltonian if and only if $B'$ is Hamiltonian (furthermore, in case they are Hamiltonian, they admit the exact same Hamiltonian cycles). The ``only if'' case is obvious since $E(B') \supseteq E(B)$. In the other direction, $B'[T]$ contains $|S|$ disjoint components, therefore a Hamiltonian cycle of $B'$ can not traverse any edge of $B'[S]$. (This is because for an arbitrary Hamiltonian cycle $H$ of $B'$ the induced graphs $H[S]$ and $H[T]$ have the same number of connected components.)

Now we form graph $B''$ from $B'$, by repeatedly applying the Bondy-Chv\'atal theorem (Lemma~\ref{bondy-chvatal}), adding edges with one endpoint in $S$ and one endpoint in $T$. It can be shown that this results in adding all edges between $S$ and $T$.

As $B''$ fully connects $S$ and $T$, it is clearly Hamiltonian, as we can traverse the $|S|$ vertex-disjoint paths in $B''[T]$ and link them together through hops via arbitrary vertices of $S$ (using each vertex in $S$ exactly once).

The claim on the size of $S$ and the fact that all possible edges between $S$ and $T$ satisfy the degree-conditions of the Bondy-Chv\'atal Theorem (Lemma~\ref{bondy-chvatal}) follow from a delicate counting argument, which crucially relies on the fact that $8n/17 < \delta < n/2$. We refer to~\cite{haggkvist} for the details. \qedd
\end{proof}

It remains to verify whether $G[T]$ can be covered by $|S|$ vertex-disjoint paths. 

\begin{lemma}\label{lempathcover}
Given an $n$-vertex graph $G$, we can find in time $O(c^t \cdot n^3)$ a cover of $G$ with $n-t$ vertex-disjoint paths, or report that no such cover exists, for arbitrary $c > (2e)^{2}$.
\end{lemma}

\begin{proof}
Apply color-coding~\cite{ColorCoding}, \cite[\S\,5.2]{book4}. Call a path nontrivial if it has more than one vertex. Say that a coloring is \emph{good} for a cover by vertex-disjoint paths, if all vertices that appear in a nontrivial path receive a different color. Clearly, if there is a cover by at most $n-t$ paths then there is a cover by exactly $n-t$ paths, and in such a cover there are at most $2t$ vertices that appear in a nontrivial path. So if there is a path cover with $n-t$ paths, a random coloring with $2t$ colors is \emph{good} for this cover with probability $e^{-2t}$. (See e.g.~\cite[Lemma 5.4]{book4}.)

On a vertex-colored graph with color set $C = \{1, \dots 2t\}$, we solve the following problem by dynamic programming: for a set $X \subseteq C$ and $v \in V(G)$, let $T[X,v]$ be the smallest number $q$ for which there exists a collection $P_1,\dots,P_q$ of vertex-disjoint paths in $G$, such that $P_q$ ends in vertex $v$ and the multiset of colors used in $P_1,\dots,P_q$ is exactly equal to $X$. (In particular, this implies that no two vertices in a path may have the same color, for it would appear twice in the multiset and only once in the set $X$.)

Let the $2t$-coloring of $G$ be given by $f : V(G) \rightarrow [2t]$. Then $T[X,v]$ satisfies the following recurrence:
\begin{itemize}
\item $T[\{c\}, v] = 1$ if $f(v) = c$,
\item $T[X, v] = +\infty$ if $f(v) \notin X$,
\item $T[X, v] = \displaystyle\min\left\{
1 + 
\min_{u \in V(G) \setminus \{v\}}{T\bigl[X \setminus \{f(v)\}, u\bigr]}, 
\min_{u \in N_G(v)}{T\bigl[X \setminus \{f(v)\}, u\bigr]}
\right\}$, otherwise.

\end{itemize}

Intuitively, the interesting part of the recurrence has two cases: either we can let $v$ be a trivial path (so we pay $1$ for having a path with $v$, and then need a collection of paths that can end at any other vertex $u$ that covers the remaining colors), or we take a system of paths covering the remaining colors that ends in a neighbor $u$ of $v$, and add the edge $uv$ to the end of that path.

Now, observe that for any color-subset $X$ and vertex $v$, there is a cover of $G$ with $T[X, v] + (n - |X|)$ paths: we cover $|X|$ vertices, one of each color in $X$, by $T[X,v]$ paths and cover the remaining $(n - |X|)$ vertices by trivial paths. So if we encounter a set $X$ and vertex $v$ for which $T[X,v] + (n - |X|) \leq n - t$, or equivalently, $T[X,v] \leq |X| - t$, then the answer is ``yes''. On the other hand, if $G$ has a path cover by $n-t$ paths and a coloring is good for this cover, then letting $X$ be the set of colors of vertices that appear in a nontrivial path and $v$ an endpoint of such a path, we obtain $T[X,v] + (n - |X|) \leq n-t$.

So by trying $e^{2t}$ random colorings and solving the dynamic program for each one, we solve the ``cover by $n-t$ disjoint paths'' problem with constant success probability. With $c^{2t}$ for $c>e$, we can boost the success probability arbitrarily close to $1$. The dynamic program can be solved in time $O(2^{2t} n^3)$. The claimed running time follows.

We may de-randomize the algorithm by replacing the randomized coloring by a deterministic construction, e.g.\ via \emph{splitters}. We omit the details of this, by now standard, technique~\cite[\S\,5.6]{book4} \qedd
\end{proof}

In our application of Lemma~\ref{lempathcover}, we need to cover $G[T]$ by $|S|$ vertex-disjoint paths. 
Observe that $|S| \geq n/2 - 3k$, and consequently $|T| \leq n/2 + 3k$.
The difference between the order of the graph $G[T]$ and the number of paths $t$ with which we want to cover it, is therefore at most $6k$.

 Applying Lemma~\ref{lempathcover}, the running time of this step is thus $O(c^{6k} \cdot n^3)$, for arbitrary $c > (2e)^2$. 


To construct a Hamiltonian cycle, find the set $S$ using Lemma~\ref{lemnw} and Lemma~\ref{lemhaggkvist}, find an appropriate path-cover using Lemma~\ref{lempathcover}, and recover the Hamiltonian cycle of $G$ by undoing the Bondy-Chv\'{a}tal steps in Lemma~\ref{lemhaggkvist}. The claimed running time of Theorem~\ref{thm2} follows by adding up the corresponding terms and by using straightforward data structuring. 

\section{Remarks and open questions}\label{sec4}
We described two algorithms that solve the Hamiltonian cycle problem, with running time that depends polynomially on the graph size and single-exponentially on the distance from Dirac's bound, a condition that guarantees the Hamiltonicity of a graph. We have considered two different ways of measuring this distance. It would be interesting to improve the bases of the exponentials in our running times, and to obtain a polynomial kernel for the second parameterization.

A natural question left open by our work is whether the two parameterizations can be combined, to obtain a generalization of both. We suspect but have not been able to prove that the following holds.

\begin{conjecture}
If at least $n-k$ vertices of $G$ have degree at least $n/2-k$, then {\sc Hamiltonian Cycle} with input $G$ can be solved in time $c^k \cdot n^{O(1)}$ for some constant $c$. 
\end{conjecture}

The results of this paper can be extended with minimal changes to similar parameterizations of Ore's theorem (Lemma~\ref{lemore}). Extending the results to generalizations of Dirac's and Ore's theorems to \emph{digraphs} would be interesting. More generally, finding new algorithms by parameterizing structural results of graph theory (whether related to Hamiltonicity or not) is a promising direction.

\section{Acknowledgement}
We thank Naomi Nishimura, Ian Goulden, and Wendy Rush for obtaining a copy of Bondy's 1980 research report~\cite{bondy}.

\newpage

\bibliographystyle{plain}
\bibliography{submission}

\end{document}